\documentclass[12pt]{article}

\usepackage{etex} % for too many package error

\usepackage[margin=3cm]{geometry}
\usepackage{amsmath}
\usepackage{amssymb}
\usepackage{amsthm}
\usepackage{fancyhdr}
\usepackage{verbatim}
\usepackage{mathtools}
\usepackage{color}
\usepackage[utf8]{inputenc}

\numberwithin{equation}{section}

\theoremstyle{plain}
\newtheorem{theorem}{Theorem}[section]
\newtheorem{definition}[theorem]{Definition}
\newtheorem{lemma}[theorem]{Lemma}

\theoremstyle{definition}
\newtheorem{remark}[theorem]{Remark}
\newtheorem{example}[theorem]{Example}

\newcommand{\Z}{\mathbb Z}
\newcommand{\N}{\mathbb N}

\begin{document} 

\title{Sturmian ground states in classical lattice-gas models} 
\author{Aernout van Enter \\ Bernoulli Institute, Nijenborgh 9 \\ Groningen University, \\9747AG, Groningen, Netherlands\\
a.c.d.van.enter@rug.nl
\\ \\ Henna Koivusalo \\ Faculty of Mathematics, University of Vienna,\\ Oskar-Morgenstern-platz 1, 1090 Vienna, Austria \\henna.koivusalo@univie.ac.at
\\ \\ Jacek Mi\c{e}kisz \\ Institute of Applied Mathematics and Mechanics \\ University of Warsaw \\ Banacha 2, 02-097 Warsaw, Poland \\ miekisz@mimuw.edu.pl} 
\pagenumbering{arabic} 

\baselineskip=20pt
\maketitle 

\begin{abstract}
We construct for the first time examples of non-frustrated, two-body, infinite-range, one-dimensional classical lattice-gas models without periodic ground-state configurations. 
Ground-state configurations of our models are Sturmian sequences defined by irrational rotations on the circle. We present minimal sets of forbidden patterns 
which define Sturmian sequences in a unique way. Our interactions assign positive energies to forbidden patterns and are equal to zero otherwise.
We illustrate our construction by the well-known example of the Fibonacci sequences.
\end{abstract}

\section{Introduction}

Since the discovery of quasicrystals \cite{shechtman}, one of the fundamental problems in statistical mechanics is to construct microscopic models of interacting atoms or molecules for which there exist thermodynamically stable, non-periodic, quasicrystalline equilibrium phases. Here we discuss one-dimensional, classical lattice-gas models without periodic ground-state configurations and with unique translation-invariant measures supported by them. In such systems, called uniquely ergodic, all (to be precise almost all) ground-state configurations locally look the same. It is known that one-dimensional systems without periodic ground-state configurations require infinite-range interactions \cite{bundangnenciu,schulradin,thirdlaw}. On the other hand, every uniquely ergodic measure is a ground-state measure of some classical lattice-gas model \cite{Aub, Rad}, but in general these might entail arbitrarily-many-body interactions. 

One-dimensional two-body interactions producing only non-periodic ground-state configurations were presented in \cite{bakbruinsma,aubry2}. Hamiltonians in these papers consisted of strictly convex two-body repelling interactions between particles and a chemical potential favoring particles. The competition between two-body interactions and the chemical potential (a source of frustration for the particles) then gives rise to what is known as a devil's staircase for the  density of particles in the ground state as a function of the chemical potential - the set of chemical potentials for which ground states have irrational density of particles is a Cantor set. 

In \cite{tmhamiltonian}, a non-frustrated, infinite range, exponentially decaying four-body Hamiltonian was constructed, with the unique ground-state-measure supported by Thue-Morse sequences. Here we present non-frustrated two-body (augmented by some finite-range interactions) Hamiltonians producing exactly the same ground states as in the frustrated model of \cite{bakbruinsma,aubry2,jedmiek1,jedmiek2}. These are the first examples of  classical-lattice gas models with such a property, the main result of this paper.

We would also like to understand what are the most important differences in the  non-periodic spatial order present in the Thue-Morse and in the Sturmian sequences, of which Fibonacci sequences are the best known examples, with respect to their stabilities.
  
To do so we discuss spatial order in one-dimensional bi-infinite sequences of two symbols, 0 and 1. The most ordered ones are of course the periodic ones. Every periodic sequence is characterized by a finite pattern, that is an assignment of symbols to a finite number {\bf p} of consecutive sites of {\bf Z}, which is repeated to the right and to the left; 
{\bf p} then  of course is the period of  given sequence. Here we are concerned with {\it non-periodic} sequences which are in some sense "most ordered" or "least non-periodic". 
Various definitions of "order" have been put forward in the mathematical literature. In particular, Sturmian systems (symbolic dynamical systems with minimal complexity)  
and balanced systems have been extensively considered, see e.g. \cite{AlSh,BaGr,fogg} and references therein. In the physics literature, most-homogeneous sequences have appeared as ground states, that is minimal-energy configurations, in certain systems of interacting particles: one-dimensional analogues of Wigner lattices \cite{hubbard}, 
the Frenkel-Kontorova model \cite{aubry1,aubry4}, the Falicov-Kimball model of itinerant electrons \cite{lemberger} where actually the term "most-homogeneous" was introduced, 
and classical lattice-gas models \cite{bakbruinsma,aubry2,aubry3,jedmiek1,jedmiek2,ultimate}. We will show here that these three notions (Sturmian, most homogeneous, balanced) 
are equivalent. We will also show that such configurations have the property of quick convergence of pattern frequencies to equilibrium values which is also called 
the strict boundary condition \cite{peyriere,strictboundary,gambaudo}. The importance of this property for stability of non-periodic ground states is discussed in \cite{strictboundary}. 

The sequences considered here give rise to uniquely ergodic 
dynamical systems. Namely, when we take any such sequence and form an infinite orbit under lattice translations, 
then the closure of this orbit supports a unique translation-invariant ergodic measure. It follows that {\em all} (rather than  almost all)  
sequences in the support of this measure look locally the same - they have the same frequencies of all finite patterns. Such systems are called {\em uniquely ergodic}. Sequences with a single defect, which are not in their orbit closure, are therefore excluded; we obtain in this way a strictly ergodic -minimal and uniquely ergodic- system. See e.g. \cite{EM2}.

In the case of configurations on $d$-dimensional lattices, $d \geq 2$, an important class of uniquely ergodic systems consists of \emph{dynamical systems 
(subshifts) of finite type} (``SOFTs''). In such systems, all configurations in the support of an ergodic measure are uniquely characterized 
by a {\em finite} family of forbidden patterns. Typical examples here are two-dimensional tiling systems \cite{robinson,shepardgrunbaum} where forbidden patterns consist 
of two neighboring square tiles with decorated edges which do not match. 

As noted before, it can be shown that one cannot have one-dimensional dynamical systems of finite type of which the support contains only non-periodic configurations \cite{bundangnenciu,schulradin,thirdlaw}. Here we show that %large class of 
Sturmian systems can be uniquely characterized by an {\em infinite} set of forbidden distances between 1's, augmented by some finite-range condition involving 0's 
(for example the absence of three consecutive 0's is part of the characterization in the case of the Fibonacci system). These are exactly the forbidden distances in the most-homogeneous description of a given Sturmian system.     

Once we find a characterization of a uniquely ergodic measure by such a ``minimal" set of forbidden patterns, we may then construct 
a relatively simple Hamiltonian which has this measure as its unique translation-invariant ground state. This implies that the configurations in its support,  
which are ground-state configurations, have minimal energy density (and moreover, we cannot decrease their local energy by a local perturbation). 
We simply assign in this construction positive energies to forbidden patterns and zero energy otherwise. 

We emphasize that our aim of getting a ``minimal" set of interactions is to have no more than two-body interactions in the infinite set of interactions we will always need. We achieve this aim, up to a single extra term. We also mention that our aim is to find out what general properties are needed from interactions to generate non-periodic order. The interaction examples we find lay no claim to being physically realistic; rather they show -and{/}or constrain- what the possibilities are. 

It is known that Sturmian sequences (most-homogeneous sequences) are ground-state configurations of frustrated interactions, as we mentioned before - repelling interactions between particles (1's in sequences) and a chemical potential favoring particles \cite{bakbruinsma,aubry2,jedmiek1,jedmiek2}. Here we construct Hamiltonians which are not frustrated and have Sturmian sequences as ground-state configurations. By combining different interaction terms in frustrated models, or by using the general results of \cite{Aub,Rad}, non-frustrated interactions might be found producing the same ground states. However, in general such constructions will not provide pair interactions. Our main new result therefore shows that in one dimension non-periodic order can occur for non-frustrated pair interactions.
   
In Section 2, we discuss various notions of order in non-periodic sequences and show their equivalence. Section 3 contains a proof that Sturmian sequences satisfy
the strict-boundary condition for all finite patterns. In Section 4 we uniquely characterize Sturmian systems (most-homogeneous configurations) 
by the absence of 1's at certain distances (augmented by the absence of some finite-range patterns). In Section 5, Sturmian systems are seen as ground states 
of certain non-frustrated Hamiltonians in classical-lattice gas models. A discussion follows in Section 6. 

Warning: As the issues we discuss have been treated in different scientific communities (e.g. ergodic theory, condensed matter physics, computer science), different terms for the same object occur. Thus an infinite Sturmian word is an infinite symbol sequence is an infinite-volume particle configuration is an infinite one-dimensional tiling, etc. Different interpretations suggest also different generalizations, such as varying the number of symbols, the dimension, etc.  As our question originated in physics (what is needed to produce non-periodic order) but the answer draws on mathematics, we will use sometimes different  terms, originating  from those different sources. We trust this will not lead to misunderstandings.

\section{Order in non-periodic sequences}

We will consider here families of bi-infinite non-periodic one-dimensional sequences of two symbols $\{ 0, 1\}$, which are such that all members of a given family look locally the same. 
Let $X \in \Omega=\{0,1\}^{\Z}$ and let $T$ be a shift operator, that is $(TX)(j)=X(j-1)$. We assume that $X$ is such that the closure (in the product topology) of the orbit 
$\{T^{i}(X), i=1,2,...\}$ supports a unique ergodic probability measure. Such a measure, $\rho$, is a limit of normalized sums of point probabilities, 

\begin{equation}\label{eq:ergodic}
\rho = \lim_{n \to \infty} \frac{1}{2n+1}\sum_{k-n \leq i \leq k+n } \delta_{T^{i}(X)},
\end{equation}
where $\delta_{T^{i}(X)}$ is a probability measure assigning probability $1$ to the configuration $T^{i}(X)$, 
and the limit is uniform with respect to $k \in \Z$.

It means that any local pattern appears with the same frequency in all sequences in the orbit closure. 
In particular, every local pattern present in $X$ appears again within a bounded distance. This property was named ``weak periodicity" in \cite{Aub}.
In Section 4, we will discuss the rate of convergence of pattern frequencies to their equilibrium values.

First we will discuss various concepts of regularity and complexity of non-periodic sequences.

\begin{definition}
The {\bf factor complexity} of an infinite word $X \in \Omega$ is the function $p_{n}$ counting the number of its {\bf factors} (finite subwords) of length $n$.
\end{definition}

It is a classical fact (see e.g. \cite{morsehedlung}) that if $p_{n} \leq n$ for some $n$, then $X$ is eventually periodic (one-way periodic beginning from some $i \in {\bf Z}$). 
It is thus the case that for each $n$ and each non-periodic word $X$ we have $p_{n} \geq n + 1$. The words with this minimal factor complexity have a special name. 

\begin{definition}
An infinite word $X$ is called {\bf Sturmian} if $p_{n} = n + 1$ for every $n$.  Taking a Sturmian word $X$, and then the closure (in the product topology) of its orbit $(T^n(X))_{n=1}^\infty$ gives a dynamical system, which we can further equip with the unique ergodic measure obtained as the limit \eqref{eq:ergodic}. We call this system the {\bf Sturmian (dynamical) system}. 
\end{definition}

Another concept of order is given in the following definition.

\begin{definition}
Denote by $|x|$ the length of a finite word $x$, and by $x(a), a=0,1$ the number of occurrences of the symbol $a$ in $x$. 
A set of words $SW$ is {\bf balanced }
if for every $x, y \in SW$ with $|x| = |y|$ one has $|x(a) - y(a)| \leq 1.$
A bi-infinite word $X \in \{0,1\}^{\Z}$ is balanced if all its factors are balanced.
\end{definition}

Balanced sequences are also called {\bf two-distance} sequences \cite{pleasants}.

We now quote the following theorem \cite[Theorem 6.1.8]{fogg}. 
\begin{theorem}\label{thm:equivalent Sturmian}
Let $X \in \{0,1\}^{\Z}$. The following conditions are equivalent:
\begin{itemize}
\item[(i)] $X$ is Sturmian and not eventually periodic 
\item[(ii)] $X$ is balanced.
\end{itemize}
\end{theorem}

Note that in the above theorem, non-periodic and 
%of
 Sturmian sequences in (i) is not enough, in view of the example of the sequence with $0$'s on negative integers 
and $1$'s on non-negative integers which is both Sturmian and non-periodic but not balanced.
\vspace{2mm}

In the physics literature \cite{hubbard,lemberger,aubry1,aubry4,bakbruinsma,aubry2,aubry3,jedmiek1,jedmiek2,ultimate} the following concept of homogeneity was considered:

\begin{definition}
Let $X \in \{0,1\}^{\Z}$ and $x_{i} \in \Z$ be the position of the $i$-th $1$ in the configuration $X$.
$X$ is {\bf most homogeneous} if there exists a sequence of natural numbers $d_{j}$ such that 
$x_{i+j} - x_{i} \in \{d_{j}, d_{j} + 1\}$ for every $i \in \Z$ and $j \in \N$. 
\end{definition}
\begin{remark}
It trivially follows that asymptotically the average distance between two particles equals $D=\lim_{j \rightarrow \infty} \frac{1}{j} d_j$. 
The ``most homogeneous" condition implies that not only  the distance between two particles with $k-1$ particles between them will be approximately  $Dk$, 
but that it will be close to that value  up to very small, bounded, fluctuations. Fluctuations of local patterns in most-homogeneous configurations are discussed in Section 3.
\end{remark}

\begin{theorem}\label{thm:equivalent homogeneous}
A sequence $X \in \{0,1\}^{\Z}$ is balanced if and only if it is most homogeneous.
\end{theorem}

\begin{proof} 

1) Let us assume that $X$ is not most homogeneous. Then we will show that it is not balanced.

It follows from the assumption that there is $j \in \N$ and two words in $X$ with 1's at their boundaries, and  $j-1$ 1's in between them,  
such that the distances between the two boundary 1's are $d_{j}$ and $d_{j} + i$ respectively,  with $i \geq 2$. 
(Notice that the lengths of these words then are $d_{j}+1$ and $d_{j} +i+1$.) Consider the following two subwords of the above words,   
of length $d_{j}+1$:
\vspace{1mm}
 
a) including the positions of two boundary 1's in the $d_{j}$ case, the number of 1's in such a word is equal to $j+1$,
\vspace{1mm}

b) excluding the positions of two boundary 1's in the $d_{j}+i$ case, the number of 1's in such a word is not bigger than $j-1$.
\vspace{1mm}
 
The numbers of 1's in these two words differ by at least 2. This shows that $X$ is not balanced.
\vspace{2mm}
 
2) Now let us assume that $X$ is not balanced and we will show that it is not most homogeneous.
\vspace{1mm}
 
Since $X$ is not most balanced, for some $n$ and $j$ there are two words of length $n$, such that there are $j$ 1's in the first word, $V$, and $j+i$, $i \geq 2$, 1's in the second word, $W$.

Firstly, we find a subword of $X$ such that it contains the word $V$, ends and begins with $1$'s, and the number of $1$'s between the first and the last $1$ is exactly $j$. (Essentially, use $10\dots 0V0\dots 01$, adding the appropriate number of $0$'s in between to make the word legal.) But then the distance between two $1$'s at the beginning and end is at least $n+1$. Hence, in the definition of most homogeneous, $d_j\ge n$.

On the other hand, consider a subword of $X$ contained in the second word $W$, beginning and ending with $1$'s which have exactly $j$ $1$'s between them. Then the distance between the beginning and the end cannot be bigger than $n-1$. This implies that in the definition of most homogeneous, $d_j \le n-1$, which is a contradiction. It follows that $X$ is not most homogeneous. 
\end{proof}

We have therefore shown that the Sturmian property is equivalent to the most-homogeneous property. We can also see the correspondence between Sturmian 
and most homogeneous systems in a direct way.

\begin{remark}\label{rem:rotations}
It is well-known (see, e.g. \cite{morsehedlung} or \cite[Theorem 10.5.8]{AlSh}) that Sturmian systems can be generated by rotations on a circle. Any such system can be associated with an irrational $\gamma<1$.
Namely, let $\psi \in [0,2\pi)$ and let $T_{\gamma}$ be a rotation on a circle by $2\pi \gamma$. We can construct a sequence $X_{\psi}$ in the following way: 
$X_{\psi}(i) = 0$ if $T_{\gamma}^{i} (\psi) \in [0, 2\pi\gamma)$, otherwise $X_{\psi}(i)=1$, for all $i \in \Z$. The closure of the orbit of $X_{\psi}$ does not depend on $\psi$ 
and it consists of Sturmian infinite words with frequency of $1's$ equal to $1-\gamma$. From now on, without loss of generality, we will assume that $\gamma>1/2$.
\end{remark}

%The collection of repetitive Sturmian words is exactly the collection of infinite words obtained in this way. (A word is said to be {\bf repetitive}, 
%if all finite subwords that occur somewhere, occur infinitely many times with bounded gaps.) The collection of factors of the Sturmian word will depend on $\gamma$ only, changing $\psi$ 
%just gives a different rotation on the circle, which gives either a translation of the same word or a word which can be approximated by such translates (it is either in the orbit, 
%or the orbit closure). 

%{\bf AvE This is ``morally, but not strictly, correct''. For a given word there are countably many translations but uncountably many rotations, that is why the orbit is not enough we need an orbit closure.}

%Denote by $[y]$ the floor of $y$, that is, the largest integer smaller than $y$. 

Let $\psi=0$. Then $X_{0}(0)=0$ and $X_{0}(1)=1$. Let us denote by $d_{j}, j=1,2, ...$, distances between $1$ at position $1$ and following $1$'s in $X_{0}$, 
that is $d_{j}$ are distances between two 1's separated by $j-1$ $1$'s. This shows that Sturmian sequences are most-homogeneous configurations with specific distances between $1$'s.

\begin{example}[Fibonacci sequences] Let us choose $\gamma$ to be equal to the reciprocal of the golden mean, $\gamma=2/(1+\sqrt{5})$. 
we choose $\psi=\gamma$, then $X_{\psi}(i), i=1,...$ is the classical Fibonacci sequence 
$0100101001001,...$ produced by the substitution rule $0\mapsto 01$, $1\mapsto 0$. 
Fibonacci sequences are all Sturmian (see, for example, \cite[Example 6.1.5]{fogg} - it follows from the fact that $11$ is a forbidden word). 
Furthermore, by Theorems \ref{thm:equivalent Sturmian} and \ref{thm:equivalent homogeneous} they are most homogeneous.

It is easy to see that here $d_{j}= [j(2+\gamma)]$, where $[y]$ denotes the floor of $y$, that is, the largest integer smaller than $y$. 
The allowed distances are therefore equal to $d_{i}$ and $d_{i}+1$, $i \in \N$.
Hence the distances $d_{j}$ are as follows: $2,5,7,10,13,15,18,20,...$. They correspond to the sequence of allowed distances $d_j, d_j+1$: $2,3,5,6,7,8,10,11, 13,14,15,16,18,$ 
$19,20,21,...$, which appear as distances between pairs of $1$'s. This leaves a list of forbidden distances: $1,4,9,12,17,22,25,...$ which never appear as distances between pairs of $1$'s. 

Let us observe that distances $d_{j}$ appear either in pairs with a difference $2$ between them or as singletons. They can be read from $X_{0}$: $X_{0}(j)=0, X_{0}(j+1)=1$ 
corresponds to the pair $(d_{j}, d_{j}+2)$ and $X_{0}(j)=0$ followed by $ X_{0}(j+1)=0$ corresponds to a singleton $d_{j}$. 
Furthermore, notice that similarly for every $j$, either $d_j-1$ or $d_j+2$ is a forbidden distance. We may also observe that there are no consecutive three 0's; 
in fact two neighboring blocks of two 0's are separated either by $1$ or by $101$. We denote by $S_{F}$ the set of all Fibonacci sequences, that is the closure of the orbit of any $X_{\psi}$.
 %We claim that all Fibonacci sequences are uniquely characterized by the absence of three consecutive 0's and the absence of pairs of 1's separated by forbidden distances.   
\end{example}

\begin{remark}\label{rem: forbidden pattern}
Inspired by the Fibonacci example, let us now analyze the allowed and forbidden distances for the general Sturmian sequences (general most-homogeneous configurations).

If $d_{1}=2$ (as in the Fibonacci system), then $d_{j}$'s appear in blocks: $d_{k}, d_{k}+2,...,d_{k}+2n$ and $d_{l}, d_{l}+2,...,d_{l}+2m$ ($|n-m|=1$) separated by one forbidden distance,
such that $d_{k}-1$, $d_{k}+2n+2$ and $d_{l}-1$, $d_{l}+2m+2$ are forbidden distances. For comparison, $n=1, m=0$ in the Fibonacci system. 

If $d_{1}>2$, then all $d_{j}$'s are singletons and $d_{j}, d_{j+1}$ are separated by $d_{1}-2$ or $d_{1}-1$ forbidden distances.
\end{remark}
   
\section{Strict boundary condition - rapid convergence of pattern frequencies} 

A frequency of a finite pattern in an infinite configuration is defined as the limit of the number of occurrences of this pattern in a segment of length $L$ 
divided by $L$ as $L \rightarrow \infty$. All sequences in any given Sturmian system have the same frequency for each pattern. We are interested now whether the fluctuations 
of the numbers of occurrences are bounded (bounded by the boundary of the size of the boundary, which in one-dimensional systems is equal to $2$). 
If that is the case, configurations are said to satisfy the {\bf strict boundary condition} \cite{strictboundary} or rapid convergence of frequencies to their equilibrium values
\cite{peyriere,gambaudo}. 

\begin{definition}
Given a sequence $X=(x_n) \in  \{0,1\}^\Z$ and a finite word $w$, define the {\bf frequency} of $w$ as 
\[
\xi_w=\lim_{N\to \infty}\frac{\#\{|n|\le N\mid x_n\dots x_{n+|w| - 1}=w\}}{2N}.
\]
Furthermore, for a segment $A\subset \Z$, denote by $X(A)$ the sub-word $(x_n)_{n\in A}$. 
We say that a sequence $X$ satisfies the {\bf strict boundary condition} (quick convergence of frequencies) if for any word $w$ and a segment $A \subset \Z$, 
the number of appearances of $w$ in $X(A)$, $n_{w}(X(A))$, satisfies the following inequality:
\[
|n_{w}(X(A)) - \xi_{w}|A|| < C_{w},
\] 
where $C_{w}>0$ is a constant which depends only on the word $w$.
\end{definition}

We will show that Sturmian sequences satisfy the strict boundary condition.
\vspace{2mm}

The following elementary fact can be found in many places in the literature. One of the earliest instances \cite{pleasants} connects balanced (or two-distance) sequences to cutting sequences, which is easily seen to be equivalent to the definition below. 
%(REFERENCE? or give a proof (not hard)):
\begin{lemma}\label{lem:components}
Let $\gamma\in (0,1)$ and $\psi \in [0, 2\pi)$, and consider the Sturmian word $X_\psi$. Denote by $\mathcal C_n$ the collection of subintervals of $[0, 2\pi)\setminus \{-k2\pi\gamma\mid k=0, \dots n\}$. 
Then the length-$n$ sub-word at the position $i$ in $X_\psi$, that is, the word $X_\psi(i)\dots X_\psi(i+n - 1)$ is uniquely determined by the subinterval $C\in \mathcal C_n$ for which $T_\gamma^i(\psi)\in C$. 
We can assume without loss of generality that the orbit of $\psi$ is never at an endpoint of an element of $\mathcal C_n$.
\end{lemma}

In other words, hitting a particular component interval is the same as seeing a particular word of length $n$. This gives us enough tools to prove the following theorem.  Results of this type have long been studied under various names. As an example, for related results in more general symbolic systems, see \cite{BeBe}, and classically \cite{Hlaw} in the context of Diophantine approximation. For completeness we provide the straightforward proof.
\begin{theorem}
Sturmian sequences satisfy the strict boundary condition. 
\end{theorem}
\begin{proof}
Let $X_\gamma(\psi)$ be Sturmian, and let $w$ be a word of length $n$. We will suppress $\psi$ in the notation below. Let $C\in \mathcal C_n$ be the component interval from Lemma \ref{lem:components} corresponding to the word $w$. Now, by Lemma \ref{lem:components} and the irrationality of $\gamma$, 
\[
\xi_w=\lim_{N\to \infty} \frac{\#\{|n|\le N\mid T^n_\gamma(\psi)\in C\}}{2N}=|C|, 
\]
where $|C|$ is the Lebesgue measure of $C$ (ergodic measure for the irrational rotation). Further, given a segment $A\subset \Z$, 
\[
n_w(X_\gamma(A))=\sum_{n\in A} \chi_C(T_\gamma^n(\psi)), 
\]
where $\chi_C$ is the characteristic function of $C$. 

It follows from Kesten's theorem
% [REFERENCE for Kesten's theorem] 
\cite{Kest}
that $C$ is a bounded remainder set; that is, it has bounded discrepancy, or
\[
|n_w(X_\gamma(A)) - |A|\xi_w| \le C_w
\]
with a constant $C_w$ that might depend on $w$. This is exactly the strict boundary condition. 
\end{proof}

\section{Forbidden-pattern characterization of Sturmian systems}

Let $(O \subset \Omega, T, \rho)$ be a uniquely ergodic dynamical system. The uniquely ergodic measure $\rho$ can be characterized by the absence of certain patterns \cite{Aub,Rad}. 
In general, the family of all forbidden patterns is rather big and it typically  consists of patterns of arbitrarily large sizes. If the family of forbidden patterns characterizing the dynamical system can be chosen to be finite, 
then we say that the corresponding dynamical system is of finite type. 

We are especially interested in uniquely ergodic measures which are non-periodic. In two dimensions, that is, for subshifts of $\{1,....,m\}^{\Z^2}$, 
non-periodic systems of finite type are given for example by non-periodic tilings by Wang tiles \cite{shepardgrunbaum,berger,robinson}. 
Forbidden patterns consist of nearest-neighbor and next-nearest-neighbor tiles that do not match.

However, it is well known that one-dimensional non-periodic systems of finite type do not exist. The proofs given in the physics literature actually show the equivalent formulation  
that any finite-range lattice-gas model with a finite one-site space has at least one periodic ground-state configuration, see for example  \cite{bundangnenciu,schulradin,thirdlaw}.

Hence, in order to uniquely characterize one-dimensional non-periodic systems we will always need to forbid infinitely many patterns. We are therefore looking for minimal families of forbidden patterns which uniquely characterize non-periodic uniquely ergodic measures. 

In the following we will be concerned with Sturmian sequences. The closure of the translation orbit of any given Sturmian sequence supports a uniquely ergodic translation-invariant probability measure. Hence it gives rise to a uniquely ergodic dynamical system called  a Sturmian system. As usual, the reader may find it helpful to keep the Fibonacci system in mind as a typical example. 

\begin{theorem}\label{thm:interactions}
Elements in any given Sturmian system are uniquely determined by the absence of the following patterns: $d_{1}+1$ consecutive $0$'s and two $1$'s separated by forbidden distances.
\end{theorem}
\begin{proof}
We first show that periodic configurations cannot satisfy the above conditions.
\vspace{3mm}

\noindent Let us note that the homogeneous configuration of just $0$'s obviously satisfies the conditions of not having the forbidden patterns of $1$'s. 
This is the reason why we need a specific finite-site condition of the absence of $0's$ which excludes such a configuration. 

Let $X \in \Omega$ be a periodic configuration (a bi-infinite sequence) with a period $p$. We will show that there is a natural number $i$ (in fact infinitely many such $i's$)  
such that $ip$ is a forbidden distance. 

We first show that there is $i$ such that $ip \neq d_{j}$ for any $j \geq 1$. 
Consider the Sturmian system on the sub-lattice $kp\Z$ of $\Z$ with $\gamma_{p} = kp\gamma \mod 1$, where $\gamma$ characterizes our original Sturmian system 
and $k$ is chosen such that $\gamma_{p}>1/2$.
Let $Y \in \{0,1\}^{pZ}$ be given by $Y(ip) = 0$ if $T_{\gamma_{p}}^{i} (\gamma) \in [0, 2\pi \gamma_{p})$, otherwise $Y(ip)=1$.
Observe that $Y(ip)=X_{0}(ip), i \geq 1$, where $X_{0}$ is the sequence generated by $T_{\gamma }(\gamma)$ (see the definition of the Sturmian systems in Remark \ref{rem:rotations}).
Obviously, there are infinitely many $0's$ in the sequence $Y(ip)$ and therefore in $X_{0}(ip)$. It means that for any such $i$, $ip \neq d_{j}$ for any $j \geq 1$.

The above argument shows more, namely that there is a natural number $i$ (in fact infinitely many such $i$'s) such that $Y(ip)=X_{0}(ip)=0$ and $X_{0}(ip-1)=0$.
For such $i's$ we have that $ip - 1 \neq d_{j}$ and therefore both $ip \neq d_{j}$ and $ip \neq d_{j}+1$ for any $j \geq 1$ hence $ip$ is a forbidden distance. 
\vspace{3mm}

\noindent Now we have to show that the only non-periodic configurations which do not have any forbidden patterns are Sturmian systems. We begin by proving that non-periodic configurations without forbidden patterns have $1$'s appearing at distances $d_j$ and $d_j+1$ for all $j$. We begin with the following lemma.
\vspace{3mm}

\begin{lemma}\label{lem:Fact1}
If two $1$'s in $X$ are at distance $d_{i}$ or $d_{i} + 1$, then there are $(i-1)$ $1$'s between them. 
\end{lemma}

\begin{proof}
This can be proved by  induction on $i$. 
The claim is immediate for $i=1$. Assume that it is true for $i$. Now consider $1$'s at a distance $d_{i+1}$,
say $X(k)=1$ and $X(k+d_{i+1})=1$. 
By the definition of the sequence $(d_j)$, we have that $d_{i+1}=d_{i}+d_{1}$ or $d_{i+1}=d_{i}+d_{1}+1$, therefore $d_{i+1}-d_{i}$ and $d_{i+1}-d_{i}-1$ 
are either forbidden distances or are equal to $d_{i}$ or $d_{i}+1$ and at least one of them is equal to $d_{i}$ or $d_{i}+1$. 
In either case there are $i$ $1$'s between $X(k)$ and $X(k+d_{i+1})$. This finishes the induction. An analogous argument can be applied in the case of two $1's$ at a distance $d_{i+1}+1$. This finishes the proof of the lemma. \end{proof}
This is used to prove the following lemma.

\begin{lemma}\label{lem:Fact2}
Any sequence $X$ which does not have any forbidden patterns has the following property:
if $X(i)=1, i\in \Z$, than for every $j \in \N$, either $X(i+d_{j})=1$ or $X(i+d_{j}+1)=1$.
\end{lemma}

\begin{proof}
If $d_{1}>2$, then $d_{j}$'s are singletons. Therefore if both $X(i+d_{j})=0$ and $X(i+d_{j}+1)=0$, 
then $X$ would have $0$'s at sites $\{i+d_{j}-(d_{1}-1),...,i+d_{j}+1+d_{1}-2\}$
or at sites $\{i+d_{j}-(d_{1}-2),...,i+d_{j}+1+d_{1}-1\}$. It would mean that $X$ has $2d_{1}-1$ successive $0$'s which is forbidden (cf. Remark \ref{rem: forbidden pattern}).

If $d_{1}=2$ and $d_{j}$'s appear in pairs or as singletons (as in the Fibonacci sequences),
then if both $X(i+d_{j})=0$ and $X(i+d_{j}+1)=0$, then $X$ would have $3$ successive $0$'s at sites
$\{i+d_{j},i+d_{j}+1,i+d_{j}+2\}$ or at sites $\{i+d_{j}-1,i+d_{j},i+d_{j}+1\}$ which is forbidden.

Now we will deal with the case when $d_{1}=2$ and $d_{j}$'s appear as blocks of size larger than $2$ (cf. Remark \ref{rem: forbidden pattern}).
Obviously if $X(i+d_{j})=0$ and $X(i+d_{j}+1)=0$ and $d_{j}$ is at the end of the block, then the argument from the previous paragraph applies. 

Hence, let us assume that $d_{j}$ is not at either end of the block $d_{\ell}, d_{\ell}+2, \dots, d_{\ell}+2n$ and further, that it is the smallest number in the sequence $(d_j)$ having the property that $X(i+d_{j})=0$ and $X(i+d_{j}+1)=0$ for some $i \in \Z$. This means that for each pair $X(i+d_{k}), X(i+d_k+1)$, with $k=\ell, \dots, j-1$ exactly one $1$ appears. Hence between $X(i+d_\ell)$ and $X(i+d_j+1)$, there are $(j-\ell)$ $1$'s. Further, to avoid the forbidden pattern of three consecutive $0$'s, it must be the case that $X(i+d_{j+1})=X(i+d_j+2)=1$. By Lemma \ref{lem:Fact1}, there should be exactly $j$ $1$'s between $X(i)$ and $X(i+d_{j+1})$. Again by Lemma \ref{lem:Fact1}, there are exactly $(\ell - 1)$ $1$'s between $X(i)$ and $X(i+d_\ell)$ or $X(i+d_\ell+1)$ (whichever of the two happens to be $1$). By the above count, this leaves only ($j-1$) $1's$ between $X(i)$ and $X(i+d_{j+1})$ (or $X(i+d_{j+1}+1)$, which is one too few, a contradiction. This ends the proof of the lemma. \end{proof}

%Then we have to have $X(i)=1, X(i+1)=0, X(i+2)=0, X(i+3)=1$, otherwise $d_{j}$ would not be the smallest distance with the above property, $d_{2}-2$ would also have this property.
%Now we look at following $1$'s to the right of $1001$. We see that after at most $\max \{n+1, m+1\}$ steps ($n+1, m+1$ are the sizes of blocks of $d_{j}$'s from Remark \ref{rem: forbidden pattern}),
%$X$ is forced to have  $X(i')=1, i' \in \Z$ and $X(i'+d_{k})=0$ or $X(i'+d_{k}+1)=0$ or  $X(i')=1, i' \in \Z$ and $X(i'-d_{k})=0$ or $X(i'-d_{k}-1)=0$ for $k<j$. 
%
%This contradicts the assumption that $d_j$ is the smallest distance. Of course, $d_{1}=2$ cannot have this property and this observation finishes the proof of the fact. 

\vspace{3mm}

\noindent By Lemma \ref{lem:Fact2}, for all $j$, at least one of $d_j$ and $d_j+1$ must repeatedly appear as a distance  between $1$'s. Further, for all $j>0$, both distances $d_j$ and $d_j+1$ must appear in $X$, otherwise (by Lemma \ref{lem:Fact1}) $X$ would be a periodic sequence and by the first part of the proof it would then have forbidden patterns. 

\vspace{3mm}

\noindent We have shown that in any $X$ which does not have forbidden patterns, any two $1$'s appear at distances $d_{j}$ or $d_{j}+1$ and in both cases there are $(j-1)$ $1's$ between them. It was proven in \cite{ultimate} that for any $0<r<1$, there exists a unique sequence $d_{j}$ such that the corresponding most-homogeneous configurations 
have $r$ as their density of $1$'s \cite[Proposition 1]{ultimate}. Furthermore, there exists a unique translation-invariant probability measure 
supported by the most-homogeneous configurations such that $r$ is the density of $1$'s \cite[Theorem 2]{ultimate}. It follows that the above-described conditions of absence of certain patterns
uniquely characterize Sturmian systems. 
\end{proof}

\section{Sturmian systems as ground states of lattice-gas models}

Once we know the set of forbidden patterns of a given symbolic uniquely ergodic dynamical system, we may construct a one-dimensional Hamiltonian 
for which the unique translation-invariant ground-state measure is given by the  uniquely ergodic measure of the corresponding dynamical system.  In particular, we have the following general statement due to Aubry (see \cite[Theorem 3]{Aub}, also see \cite{Rad}). 
\begin{theorem}[Aubry \cite{Aub}]\label{thm:Aubry}
For any weakly periodic configuration of (pseudo-)spins on a cubic lattice, there exists a well-defined Hamiltonian for which the set of ground states 
is identical to the closed orbit of this configuration under the translation group $\mathbb Z^d$. 
\end{theorem}
In our setting, it suffices to say that a configuration of (pseudo-)spins on a cubic  --here linear-- lattice is an infinite word $X\in \{0,1\}^{\mathbb Z}$. {\it Weakly periodic} means that for any finite word $B$ appearing in $X$there is a number $N$ such that any word of length $N$ appearing in $X$ contains $B$ as a subword. 

We have the following theorem.

\begin{theorem}\label{thm:full}
For every Sturmian system there exists a one-dimensional, non-frustrated, arbitrarily fast decaying, lattice-gas (essentially) two-body Hamiltonian (augmented by some finite-range non-frustrated interactions) 
for which the unique ergodic translation-invariant ground-state measure is the ergodic measure of the Sturmian system.
\end{theorem}

\begin{proof}
Sturmian words are weakly periodic (Sturmian words are known to be {\it repetitive}, see \cite{BaGr}), so that Theorem \ref{thm:Aubry} applies. The proof in \cite{Aub} is constructive, and in particular, it can be gleaned that by Theorem \ref{thm:interactions} for the Sturmian systems, the Hamiltonian simply penalizes the forbidden patterns, that is it assigns to them positive energies, while the energy of all other patterns is equal to zero. 

The construction is as follows. For distances $d_{j}, d_{j}+1$, the pair-interaction energy between two particles ($1$'s) is zero, otherwise it is positive. 
Moreover we forbid $d_{1}+1$ successive $0$'s. 

So we have a lattice-gas model with a finite-range term (a positive energy assigned to $d_{1}+1$ successive $0$'s) plus pair interactions 
$\sum_{i,j \in \Z}  J(j)  n_{i} n_{i+j}$ where $J(j)>0$ is a coupling constant which may decay at infinity arbitrarily fast, $n_ {i} =1$ if the lattice site $i$ is occupied;
that is, we have $1$ at a corresponding Sturmian sequence at site $i$.  

The final statement on the ground-state measure follows from the fact that the Sturmian system is uniquely ergodic. 
\end{proof}

We end this section with a comparison of the above theorem to relevant related results in the literature and discussion on directions for future work. 

To begin the discussion, we mention a similar result that holds for the Thue-Morse system. A non-periodic Thue-Morse sequence is produced by the substitution rule 
$0\mapsto 01$, $1\mapsto 10$, and is a canonical example of a one-dimensional aperiodic pattern. It was shown in \cite{gothed, gotheda} 
that the Thue-Morse system is uniquely characterized by the absence of the following forbidden patterns:
$BBb$, where $B$ is any word and $b$ is its first letter. In \cite{tmhamiltonian}, a minimal set of forbidden patterns which involve only $4$ lattice sites 
at specific distances was found. This allowed the construction of a 4-body Hamiltonian with exponentially (or even faster) decaying interactions 
for which the Thue-Morse sequences are the only ground-state configurations.

However, the above result is in stark contrast to the two-dimensional case. Namely, for two-dimensional systems of finite type, 
the above construction gives us a classical lattice-gas model with finite-range interactions, but it was shown in \cite{ultimate} 
that the reverse statement is not true in general: A classical lattice-gas model with finite-range interactions
was constructed with the property that its uniquely ergodic ground-state measure is not equal to any ergodic measure of a dynamical system of finite type. 
In fact uncountably many such classical lattice-gas models were constructed with ground state-measures given by two-dimensional analogues of Sturmian systems. 
There are only countably many systems of finite type which shows that the family of ergodic ground-state measures of finite-range lattice-gas models is much larger than
the family of ergodic measures of dynamical systems of finite type. 

Classical lattice-gas models corresponding to systems of finite type based on Robinson's non-periodic tilings were the first examples of systems of interacting particles 
without periodic ground-state configurations - microscopic models of quasicrystals \cite{Rad0,mr,cmpmiekisz,strictboundary}.   

The case of Sturmian systems has also been discussed in earlier works. One-dimensional Hamiltonians with infinite-range, exponentially decaying, convex, repulsive interactions, 
and a chemical potential favoring the presence of particles, were studied in \cite{bakbruinsma,aubry2}. It was shown that the density of particles in the ground state 
as a function of the chemical potential is given by a devil's staircase, that is it has the structure of a Cantor set. Let us note that the Hamiltonian in \cite{bakbruinsma, aubry2}
is frustrated, so that ground-state configurations arise as a result of the competition between repelling interactions and a chemical potential. In Theorem \ref{thm:full}, 
in contrast, we constructed non-frustrated Hamiltonians for most-homogeneous configurations, therefore for Sturmian systems. 

Another key property from the perspective of physical interpretations of non-periodic patterns is the stability of the pattern under perturbations. 
It was shown in \cite{strictboundary} that the strict boundary condition is equivalent to zero-temperature stability of two-dimensional non-periodic ground states of classical-lattice
gas models. More precisely, non-periodic ground states are stable against small perturbations of the range $r$ if and only if the strict boundary condition is satisfied 
for all local patterns of sizes smaller than $r$. We conjecture that the strict boundary condition is equivalent to low-temperature stability of non-periodic
ground states, that is to the existence of non-periodic Gibbs states.

The situation is much more subtle in  models with infinite-range interactions, whether in one or in more dimensions. In one dimension, 
non-periodic ground states are obviously not stable against interaction perturbations in which the tail is cut off so that the perturbed interaction is finite-range, 
as then at the least new periodic ground states will arise.

Moreover, perturbing any coexistence of ground states or Gibbs measures in any dimension with an interaction with  
an arbitrarily small $l_{1}$ norm can cause instabilities (see e.g. \cite{DvE, MR89}), which indicates that the interaction spaces with $l_1$-like norm may be too large. 
Also, existence statements for interactions with such a finite $l_1$ norm, having prescribed long-range order properties,  can be derived via the Israel-Bishop-Phelps theorem \cite{Isr, EM,EZ}. 
In particular in \cite{EZ}  Sturmian-like long-range order is derived for long-range pair interactions. However, beyond there being no control on the long-range behaviour of the interactions, 
the interactions obtained by this method are not frustration-free, and neither can we say much about uniqueness of the translation-invariant Sturmian ground states or Sturmian-like Gibbs measures. 

Another pertinent observation is that if the interactions are sufficiently many-body and long-range, non-periodic ground states can be stable even at positive temperatures 
(freezing transitions may occur) \cite{BruLep1, BruLep2}.

Thus the appropriate stability properties of Sturmian, as well as  more general non-periodic, ground states are still a matter about which our knowledge is insufficient.
%- such a perturbation transforms a given interaction into a finite-range one. It seems thatthe stability of one-dimensional non-periodic configurations against finite-range interactions depends on the rate of the decay of the interaction at infinity.

\section{Discussion}

We have discussed various notions of complexity and order in non-periodic one-dimensional sequences (lattice configurations), in particular Sturmian systems, balanced sequences, and most-homogeneous sequences. 
We have shown that all these notions of ``almost" periodicity are equivalent.

Our main result is that most-homogeneous sequences (Sturmian sequences) are uniquely characterized by the absence of pairs of 1's at certain distances (augmented by the absence of finite patterns, 
such as the absence of three consecutive 0's in the Fibonacci system). This then allowed us to construct one-dimensional lattice-gas models with exponentially decaying two-body interactions 
which have a given Sturmian ergodic measure as a unique ground-state measure. Our result provides the first examples of non-frustrated essentially two-body Hamiltonians without periodic ground-state configurations.

It is a highly interesting but challenging question to see if we can find conditions which cause such one-dimensional non-periodic
ground states to be stable in some sense; for example are they thermodynamically stable at sufficiently low but non-zero temperatures, that is, 
do they give rise to non-periodic Gibbs states, either by adding extra dimensions in which ferromagnetic couplings are present, as in \cite{EMZ}, 
or by adding some explicit, sufficiently long-range, interactions? Or can we say that they are stable at $T=0$ , as discussed in \cite{Miesta}? 

Short-range interactions in one dimension can never have ordered Gibbs states, so the stability can either be at $T=0$, or will necessarily require long-range interactions or extra dimensions.   

{\bf Acknowledgments} JM and AvE would like to thank the National Science Centre (Poland) for financial support under Grant No. 2016/22/M/ST1/00536. 
HK gratefully acknowledges the support of OeAD grant number PL03/2017.
JM thanks Karol Penson for introducing to him a wonderful word of the On-Line Encyclopedia of Integer Sequences and Marek Biskup for many helpful discussions.

\end{document}